\documentclass[11pt]{article}
 \usepackage{amsmath}
 \usepackage{amsfonts}
 \usepackage{amssymb}
 \usepackage{amsthm}
 
\usepackage{a4wide}

\theoremstyle{plain}
\newtheorem{theorem}{Theorem}
\newtheorem{lemma}{Lemma}
\newtheorem{proposition}{Proposition}
\newtheorem{corollary}{Corollary}

\theoremstyle{definition}
\newtheorem{definition}{Definition}
\newtheorem{example}{Example}

\theoremstyle{remark}
\newtheorem{remark}{Remark}
 
\newcommand{\LD}{\mathord{\backslash}}
\newcommand{\RD}{\mathord{\slash}}
\newcommand{\PR}{\bullet}
\newcommand{\LDI}{\LD\!\!\LD}
\newcommand{\RDI}{\RD\!\!\RD}
\let\phi\varphi

\begin{document}
\title{Iterative division in the Distributive Full Non-associative Lambek Calculus\thanks{This is a preprint of an article to appear in the proceedings of the 2nd DaL\'i Workshop, Dynamic Logic: New Trends and Applications, Porto 2019, to be published by Springer. Please use and cite the publisher's version. This work is supported by the Czech Science Foundation grant number GJ18-19162Y. The author is grateful to Andrew Tedder and three anonymous reviewers for useful comments.}
}
\author{Igor Sedl\'ar\\[2mm]
The Czech Academy of Sciences, Institute of Computer Science\\
Pod Vod\'arenskou v\v{e}z\'i 271/2, Prague, The Czech Republic\\
\texttt{sedlar@cs.cas.cz}}

\maketitle              
\begin{abstract}
We study an extension of the Distributive Full Non-associative Lambek Calculus with iterative division operators. The iterative operators can be seen as representing iterative composition of linguistic resources or of actions. A complete axiomatization of the logic is provided and decidability is established via a proof of the finite model property.

\end{abstract}

\section{Introduction}
Operators of the Lambek Calculus \cite{Lambek1958} were designed to represent types of linguistic expressions. The ``product'' operator $\PR$ articulates the internal structure of expressions---an expression is of type $A \PR B$ if it is a result of \emph{concatenating} an expression of type $A$ with an expression of type $B$, in that order. The left and right ``division'' operators, $\LD$ and $\RD$, describe the behaviour of expressions under concatenations---an expression $x$ is of type $A \LD B$ if, for all expressions $y$ of type $A$, the concatenation $yx$ is of type $B$; similarly for $B \RD A$ and $xy$. The interpretation of product in terms of concatenation of strings requires the product operation to be associative, which also entails some particular properties of the division operators.

The Lambek operators admit more general interpretations as well. For instance, we may read them in terms of  \emph{merging pieces of information} or \emph{composing linguistic resources}; see \cite[p.\ 350]{Moortgat1996} or \cite[p.\ 10]{Kurtonina1994}. On this interpretation, $A \LD B$ denotes resources $x$ such that if $x$ is composed with an ``input'' resource of type $A$, then the result is a resource of type $B$; similarly $B \RD A$ denotes resources $x$ such that any composition of a resource of type $A$ with ``input'' $x$ will result in a resource of type $B$. Adopting the perspective of arrow logic \cite{Benthem1995,Venema1997}, the division operators can also be seen as describing \emph{composition of actions}. On this interpretation, $A \LD B$ denotes actions $x$ such that any action consisting of performing an action of type $A$ and then performing $x$ is an action of type $B$; $B \RD A$ denotes actions $x$ such that  performing $x$ and then performing an action of type $A$ is always of type $B$.

Assuming either of these interpretations, it is natural to consider \emph{iterated composition} in addition to ``one-off'' composition represented by the Lambek division operators. As an example of iterated composition, take a linguistic resource $y$ and an arbitrary non-empty finite sequence $x_1, \ldots, x_n$ of resources of type $A$, and compose $y$ with input $x_1$, then compose the result with input $x_2$ and so on. With an \emph{iterative} left division operator at hand, say $\LDI$, we could denote by  $A \LDI B$ resources $y$ such that iterated composition of $y$ with any non-empty sequence of resources $x_1, \ldots, x_n$ of type $A$, as indicated above, is guaranteed to result in a resource of type $B$. As another example, take an action $y$ and a finite non-empty sequence $x_1, \ldots, x_n$ of actions of type $A$. Perform $x_1$ after $y$, then perform $x_2$ and so on. With an iterative right division operator at hand, say $\RDI$, we could denote by $B \RDI A$ actions $y$ such that performing any sequence of actions of type $A$ after $y$, as indicated above, is guaranteed to result in an action of type $B$. Iterative composition can be seen as a generalization of action iteration in dynamic logic; the former corresponds to a binary modality while the latter to a unary one. Accordingly, the question concerning the result of iterated composition can be seen as a question concerning ``correctness'' of iterative composition with respect to specific postconditions.

\begin{example}
Suppose we want to know if the core beliefs of some agent, say Ann, concerning some topic of interest, are ``immune'' to fake news concerning the topic. We can rephrase this by asking if exposing Ann consecutively to any number of fake news related to the topic changes her core beliefs concerning the topic. This is a question concerning iterated composition---let $y$ represent Ann's initial belief state, let $A$ denote fake news concerning the given topic and let $B$ denote belief states that support Ann's initial core beliefs concerning the topic; the question is whether combining $y$ consecutively with any $x_1, \ldots, x_n$ of type $A$ results in a belief state of type $B$.
\end{example}

\begin{example}
Suppose we have programmed a robot to perform certain tasks under specific observed conditions, for example re-arranging objects in a warehouse. For security reasons, we want to know if any ``legal'' sequence of the robot's actions can lead to a situation endangering the human workers in the warehouse. Again, we have here an issue pertaining to iterated composition---the question is whether, given some initial action $y$, all sequences of actions of the type ``legal'' performed consecutively after $y$ constitute an action of type ``not endangering the human workers''.
\end{example}

Note that $A \LDI B$, interpreted as above, intuitively corresponds to an infinite conjunction of formulas $A \LD B$, $A \LD (A \LD B)$, $A \LD (A \LD (A \LD B))$ and so on; similarly $B \RDI A$ corresponds to an infinite conjunction of $B \RD A$, $(B \RD A) \RD A$ and so on. On the assumption that composition (product) is associative, the former series of formulas is equivalent to $A \LD B$, $(A \PR A) \LD B$, $\big( A \PR (A \PR A ) \big ) \LD B$ and so on; and similarly for the latter series of formulas. Hence, on the assumption of associativity, iterative division of both kinds can be formalised using a transitive-only version of the ``continuous'' Kleene star opertor familiar from regular languages; see especially Pratt's Action Logic \cite{Pratt1991}, Kozen's $*$-cointinuous action lattices \cite{Kozen1994a} or van Benthem's Dynamic Arrow Logic \cite{Benthem1994}. As ``the Kleene plus'' of $A$, $A^{+}$, intuitively represents $A \lor (A \PR A) \lor \big (A \PR (A \PR A) \big ) \lor \ldots$, ``repeated considerations of information of type $A$ resulting in $B$'' (the first example), i.e.\ $A \LDI B$, can be formalized as $A^{+} \LD B$ and ``repeating $A$-type actions with result $B$'' (the second example), i.e.\ $B \RDI A$, as $B \RD A^{+}$.

However, the assumption of associativity, necessary for the reduction of iterative division to the Kleene operator, is problematic for at least two reasons. Firstly, it has been shown that associativity often leads to undecidability. Buszkowski \cite{Buszkowski2006a} and Palka \cite{Palka2007} have shown that the logic of all $*$-continuous action lattices (the ones where, roughly, Kleene star is equivalent to an infinite disjunction) is undecidable; Kuznetsov \cite{Kuznetsov2019} established recently that the logic of all action lattices is undecidable as well. His proof applies also to Pratt's Action Logic. It follows from the results of \cite{Kurucz1995} that Associative Dynamic Arrow Logic is undecidable. 

Secondly, associativity of composition is problematic on some information-related and channel-theoretic interpretations of the Lambek Calculus; see e.g.\ \cite{Sequoiah2013,Tedder2017} for more details. Lambek \cite{Lambek1961} provides motivation for getting rid of associativity even on the original linguistic interpretation of his calculus. (On this interpretation, the object of study are ``bracketed strings'' over an alphabet---in presence of associativity, brackets can be omitted.)

Thus a natural question arises concerning a formalization of iterative composition in the more general non-associative setting. In this paper we extend the Distributive Non-associative Lambek Calculus with primitive iterative division operators $\LDI$ and $\RDI$, representing the two kinds of iterative composition outlined in the above examples. Whe show that the resulting logic is decidable and we provide a sound and complete axiomatization for it. Our starting point is the relational semantics for the Distributive Non-associative Lambek Calculus in the style of Do\v{s}en \cite{Dosen1992} using a ternary accessibility relation. In Section \ref{sec: models} we use this semantics to give satisfaction clauses for formulas with $\LDI$ and $\RDI$. In Section \ref{sec: comp} we provide a weakly complete axiomatization of the theory of all such models and prove that the theory is decidable. In the concluding Section \ref{sec: conclusion} we discuss variants and extensions of our basic logic and point out some possible directions of future work. 

\textit{Related work.} Bimb\'o and Dunn \cite{Bimbo2005} provide a relational semantics for Pratt's Action Logic (a logic containing versions of the Lambek division operators along with the Kleene star) and for the related logic of Kleene Algebras \cite{Kozen1994}. Non-associative versions of Kleene algebra were studied in \cite{Desharnais2017} with a motivation coming from temporal logic. The iterative division operators studied in the present paper seem to be new, although definable in existing frameworks on the assumption of associativity.

\section{Relational models and iteration}\label{sec: models}

A \emph{Do\v{s}en frame} is $F = \langle S, R\rangle$, where $S$ is a non-empty set (``states'') and $R$ is a ternary relation on $S$. States can be seen, for instance, as linguistic resources. On that interpretation, $Rstu$ means that composing resource $t$ with input $s$ might result in resource $u$ (composition is in general non-deterministic). States can also be seen as actions, in which case $Rstu$ is taken to mean that action $u$ may be decomposed into $s$ followed by $t$ (as in arrow logic) or, more generally, as saying that performing $t$ after $s$ may give $u$.

Fix a denumerable set $Pr$ of propositional variables, intuitively representing some basic features of states. A \emph{Do\v{s}en model} is $M = \langle S, R, V\rangle$ where $V$ is a function from $Pr$ to subsets of $S$; $V(p)$ is the set of states having feature $p$. We say that $\langle S, R, V\rangle$ is a model based on frame $\langle S, R\rangle$.

Various formal languages can be used to express complex features of states. Our basic language is the Lambek language with bounded-lattice connectives $\mathcal{L}_{[\LD, \PR, \RD, \land, \lor, \top, \bot]}$, concisely denoted as $\mathcal{L}_0$, containing zero-ary operators $\top$ (``top'', ``truth'') and $\bot$ (``bottom'', ``falsity'') and binary operators $\LD$ (``left division''), $\PR$ (``product'', ``fusion''), $\RD$ (``right division''), $\land$ (``meet'', ``conjunction'') and $\lor$ (``join'', ``disjunction''). The set of \emph{$\mathcal{L}_0$-formulas} (over $Pr$) is defined as usual. \emph{Sequents} are ordered pairs of formulas. For each $M$, we define the \emph{satisfaction relation $\vDash_M$} between states of the model and formulas as follows:
\begin{itemize}
\item[] $s \vDash_M p$ iff $s \in V(p)$;
\item[] $s \vDash_M \top$ for all $s$; $s \not\vDash_M \bot$ for all $s$;
\item[] $s \vDash_M A \land B$ iff $s \vDash_M A$ and $s \vDash_M B$ ;
\item[] $s \vDash_M A \lor B$ iff $s \vDash_M A$ or $s \vDash_M B$;
\item[] $s \vDash_M A \PR B$ iff there are $t,u$ such that $Rtus$ and $t \vDash_M A$ and $u \vDash_M B$;
\item[] $s \vDash_M A \LD B$ iff, for all $t$ and $u$, if $Rtsu$ and $t \vDash_M A$, then $u \vDash_M B$;
\item[] $s \vDash_M B \RD A$ iff, for all $t$ and $u$, if $Rstu$ and $t \vDash_M A$, then $u \vDash_M B$.
\end{itemize}
A sequent $A \vdash B$ is \emph{valid in $M$} iff $s \vDash_M A$ implies $s \vDash_M B$ for all states $s$ of $M$; $A \vdash B$ is \emph{valid in a frame $F$} iff it is valid in all models based on the frame. For each language $\mathcal{L}$ considered in this paper, the \emph{$\mathcal{L}$-theory} of a class of frames is the set of sequents of $\mathcal{L}$-formulas valid in each frame belonging to the class.

Informally, formulas express types of states; $s \vDash_M A$ is read as ``$s$ is of type $A$ (in $M$)''. Sequents express type dependency; $A \vdash B$ is valid in $M$ iff all states of type $A$ are of type $B$ (in $M$). 

We use the standard \cite{Restall2000} notation
\begin{align*}
Rstuv &:= \exists x (Rstx \And Rxuv) &&
Rs(tu)v := \exists x (Rsxv \And Rtux)
\end{align*}
Let $\bar{x}$ be a finite non-empty list of states $\langle x_1, \ldots, x_n\rangle$, which we call a ``path of length $n$''. We define
\begin{align*}
R\overleftarrow{x}st & :=
\exists y_1, \ldots, y_{n-1} \big ( 
Rx_1sy_1 \land \Big ( \bigwedge_{1 \leq i \leq n-2} Rx_{i+1} y_i y_{i+1} \Big) \land Rx_{n}y_{n-1}t
\big ) \\
R s \overrightarrow{x} t & := 
\exists y_1, \ldots, y_{n-1} \big ( 
Rsx_1y_1 \land \Big ( \bigwedge_{1 \leq i \leq n-2} R y_i x_{i+1} y_{i+1} \Big ) \land Ry_{n-1} x_{n} t
\big )
\end{align*}
(Hence, if $\bar{x} = \langle x_1, x_2\rangle$, then $R \overleftarrow{x}st$ is $\exists y (R x_1 s y \land R x_2 y t )$, i.e.\ $R x_2 (x_1 s) t$, whereas $R s \overrightarrow{x} t$ is $\exists y (R s x_1 y \land R y x_2 t )$, i.e.\ $R s x_1 x_2 t$ .)

We say that $\bar{x} = \langle x_1, \ldots, x_n\rangle$ satisfies $A$ in $M$, notation $\bar{x} \vDash_M A$, iff $x_i \vDash_M A$ for all $i \in \{ 1, \ldots, n \}$. We define
\[ 
A \LD^{\! 1} B := A \LD B \quad 
A \LD^{\! n+1} B := A \LD (A \LD^{\! n} B)  \quad
B \RD^{1} A := B \RD A \quad 
B \RD^{n+1} A := (B \RD^{n} A) \RD A
\]
Note that $A \LD^{\! n} B$ denotes a type of object such that, if $n$ inputs of type $A$ are consecutively combined with the given object, then the result will be of type $B$; $B \RD^{n} A$ has a similar meaning. The iterative division operators we are after can be seen, informally, as corresponding to infinite conjunctions:
\[
\bigwedge_{ 1 \leq n } A \LD^{\! n} B \qquad
\bigwedge_{ 1 \leq n } B \RD^{n} A
\]
Of course, these are not really formulas of our language as conjunction are finite. Hence, we need to express iterative division in some other way.

\begin{lemma}\label{lem: subpath}
Let $\bar{x} = \langle x_1, \bar{z}\rangle$. Then
\begin{itemize}
\item[(a)] $\bar{x} \vDash_M A$ only if $\bar{z} \vDash A$;
\item[(b)] $R \overleftarrow{x} st$ only if $R x_1 s y_1$ and $R \overleftarrow{z} x_1 t$ for some $y_1$;
\item[(c)] $R s \overrightarrow{x} t$ only if $R s x_1 y_1$ and $R x_1 \overleftarrow{z} t$ for some $y_1$.
\end{itemize}
\end{lemma}
\begin{proof}
Item (a) is obvious. (b) follows from the assumption that $z_i = x_{i+1}$ and the definition of $R s \overrightarrow{x} t$. (c) is established similarly.
\end{proof}

\begin{proposition} \mbox{}
\begin{itemize}
\item[(a)] $s \vDash_M A \LD^{\! n} B$ iff, for all paths $ \bar{x} $ of length $n$ and all $t$, if $R \overleftarrow{x} st$ and $\bar{x} \vDash_M A$, then $t \vDash_M B$.
\item[(b)] $s \vDash_M B \RD^{n} A$ iff, for all all paths $ \bar{x} $ of length $n$ and all $t$, if $R s \overrightarrow{x} t$ and $\bar{x} \vDash_M A$, then $t \vDash_M B$.
\end{itemize}
\end{proposition}
\begin{proof}
Induction on $n$. Both base cases $n = 1$ are straightforward consequences of the satisfaction clauses for the division operators. Next, we show that if (a) holds for $k \in \{ 1, \ldots, n - 1 \}$, then it holds for $k+1$; a similar claim about (b) is established analogously. Firstly, assume that $s \vDash A \LD^{\! k+1} B$ and we have $\bar{x} = \langle x_1, \ldots, x_{k+1}\rangle = \langle x_1, \bar{z}\rangle$ such that $R \overleftarrow{x} s t$ and $\bar{x} \vDash A$. We have to show that $t \vDash B$. By the definition of $A \LD^{\! k+1} B$, we have $s \vDash A \LD (A \LD^{\! k} B)$ and by the definition of $R \overleftarrow{x} st$ we have $Rx_1 s y_1$ and $R \overleftarrow{z} y_1 t$ for some $y_1$. Hence, $y_1 \vDash A \LD^{\! k} B$ by the satisfaction clause for left division and so $t \vDash B$ by the induction hypothesis ($\bar{z}$ is obviously a path of length $k$). Secondly, we show that if $s \not\vDash A \LD^{\! k+1} B$, then there is a path $\bar{x}$ of length $k+1$ such that $R \overleftarrow{x} s t$, $\bar{x} \vDash A$ and $t \not\vDash B$. If $s \not\vDash A \LD^{\! k+1} B$, then $s \not\vDash A \LD (A \LD^{\! n} B)$ by the definition of $A \LD^{\! k+1} B$. Hence, there are $x,y$ such that $Rxsy$, $x \vDash A$ and $y \not\vDash A \LD^{\! k} B$. The latter means, by the induction hypothesis, that there is a sequence $\bar{z}$ of length $k$ and a state $t$ such that $R\overleftarrow{z}yt$ and $\bar{z} \vDash A$ while $t \not\vDash B$. But $\bar{x} = \langle x, \bar{z}\rangle$ is a sequence of length $k+1$ satisfying $A$ such that $R\overleftarrow{z} s t$.
\end{proof}

\noindent The proposition provides a lead as to the appropriate satisfaction conditions for iterated versions of the division operators to which we now turn.  

The Lambek language with distributive bounded-lattice connectives and iterative division operators $\mathcal{L}_{[\LD, \PR, \RD, \LDI, \RDI, \land, \lor, \top, \bot]}$, denoted also as $\mathcal{L}_{1}$, adds to $\mathcal{L}_0$ two binary operators $\LDI$ (``iterative left division'') and $\RDI$ (``iterative right division''). A \emph{finite path} of elements of a model is a path $\bar{x}$ of length $n$ for some natural number $n \geq 1$. The satisfaction relation, when extended to $\mathcal{L}_1$, is assumed to satisfy the following new clauses:
\begin{itemize}
 \item[] $s \vDash_M A \LDI B$ iff, for all $t$ and finite paths $\bar{x}$, if $R \overleftarrow{x} st$ and $\bar{x} \vDash_M A$, then $t \vDash_M B$; 
  \item[] $s \vDash_M B \RDI A$ iff, for all $t$ and finite paths $\bar{x}$, if $R s \overrightarrow{x} t$ and $\bar{x} \vDash_M A$, then $t \vDash_M B$.
 \end{itemize} 
Validity is defined as before.

\begin{proposition}
The following sequents are valid in all Do\v{s}en frames:
\begin{itemize}
	\item[(a)] $A \LDI B \land A \LDI C \vdash A \LDI (B \land C)$ and $B \RDI A \land C \RDI A \vdash (B \land C) \RDI A$;
	\item[(b)] $A \LDI B \vdash A \LD B \land A \LD (A \LDI B)$ and $B \RDI A \vdash B \RD A \land (B \RDI A) \RD A$;
	\item[(c)] $A \LD B \land A \LD (A \LDI B) \vdash A \LDI B$ and $(B \RDI A) \RD A \land B \RD A \vdash B \RDI A$.
\end{itemize}
The following rules preserve validity in Do\v{s}en models:
\begin{itemize}
	\item[(d)] $\dfrac{A \vdash B \quad C \vdash D}{B \LDI C \vdash A \LDI D}$ and
	$\dfrac{A \vdash B \quad C \vdash D}{C \RDI B \vdash D \RDI A}$;
	\item[(e)] $\dfrac{A \vdash B \LD A}{A \vdash B \LDI A}$ and
	$\dfrac{A \vdash A \RD B}{A \vdash A \RDI B}$.
\end{itemize}
\end{proposition}
\begin{proof}
We show just the $\LDI$-parts of (b), (c) and (e). (b) Firstly, it is clear that $A \LDI B \vdash A \LD B$ is valid (consider paths of length $1$). Secondly, fix $s$, take some $ t, u$ and assume that there are $x$ and $\bar{z}$ satisfying $A$ such that $Rxst$ and $R \overleftarrow{z} t u$. We have to prove that $u \vDash B$. It is clear that $R \overleftarrow{\langle x, \bar{z}\rangle} s u$ and $\langle x, \bar{z}\rangle \vDash A$; hence if $s \vDash A \LDI B$, then $u \vDash B$ and so, in general, $s \vDash A \LD (A \LDI B)$. 

(c)  Assume that $s \vDash A \LD (A \LDI B) \land A \LD B$ and take some finite path $\bar{x}$ satisfying $A$ such that $R \overleftarrow{x} s t$. We have to show that $t \vDash B$. If $\bar{x} = \langle x\rangle$, then this follows from the assumption that $s \vDash A \LD B$. If $\bar{x} = \langle x, \bar{z}\rangle$, then $R \overleftarrow{z} y t$ and $y \vDash A \LDI B$ by Lemma \ref{lem: subpath} and the assumption $s \vDash A \LD (A \LDI B)$; it then follows readily that $t \vDash B$.

To prove (e), assume that $s \vDash_M A$ and $R \overleftarrow{x} s t$ for some $\bar{x} \vDash_M B$ and $t$. We have to show that $t \vDash_M A$. By the definition of $R\overleftarrow{x} s t$, there are $y_1, \ldots, y_{n-1}$ such that 
\[
Rx_1sy_1 \land \bigwedge_{1 \leq i \leq n-2} Rx_{i+1} y_i y_{i+1} \land Rx_{n}y_{n-1}t
\]
If $A$ entails $B \LD A$ in $M$, then $y_i \vDash_M A$ for all $i \in \{ 1, \ldots, n-1 \}$. Hence, $y_{n-1} \vDash_M B \LD A$ and so $t \vDash_M A$. 
 \end{proof}
 
 \begin{remark}\label{rem: transitive closure}
 We may define the \emph{left transitive closure} $R^{+l}$ of $R$ stepwise as follows:
 \begin{align*}
 R^{1} & := R \\
 R^{n+1} & := \{ \langle s, t, u \rangle \mid \exists x y ( R s y u \land R^{n} x t y ) \} \\
 R^{+l} & := \bigcup_{1 \leq n} R^{n}
 \end{align*}
 The right transitive closure $R^{+r}$ of $R$ may be defined similarly. It is interesting to note that the implication 
 \begin{equation}\label{eq: transitive closure}
 s \vDash_M A \LDI B \implies  \Big ( \forall t, u (R^{+l} tsu \land t \vDash_M A \Rightarrow u \vDash_M B \Big)
 \end{equation}
 is not valid in general, although the converse implication is valid. (A similar claim holds for $R^{+r}$ and $B \RDI A$.) To see the former, take a model where $R = \{ \langle x, s, y\rangle, \langle t, y, u\rangle\}$ and $V(p) = \{ t \}, V(q) = \emptyset$. Then $R^{+l} = \{ \langle t, s, u \rangle \}$, so the consequent of \eqref{eq: transitive closure} fails for $A = p$ and $B = q$. However, there are no finite paths $\bar{x}$ such that $R \overleftarrow{x} s u$ and $\bar{x} \vDash A$ (only $\langle x, t\rangle$ satisfies the former condition, but $p$ fails in $x$), so vacuously $s \vDash p \LDI q  $. To see the latter, observe that if $R\overleftarrow{x}su$ and $x$ is the last element of $\bar{x}$, then $R^{+l}xsu$. 
 \end{remark}
 
 We note that \[ A \LDI C \land B \LDI C \vdash (A \lor B) \LDI C 
 \quad\text{ and }\quad
C \RDI A \land C \RDI B \vdash C \RDI (A \lor B) 
 \] are not valid; the reader is invited to find a counterexample as an exercise.

\section{Completeness and decidability}\label{sec: comp}
In this section we provide a sound and (weakly) complete axiomatization of the $\mathcal{L}_1$-theory of all Do\v{s}en frames and we show the theory to be decidable via a finite canonical model construction. Our technique derives from \cite{Nishimura1982}; it is a variant of the finite canonical model construction often used in completeness proofs for logics with fixpoint operators such as epistemic logics with common knowledge or Propositional Dynamic Logic. 

Let us consider the following axiom system, denoted as $\mathit{IDFNL}$:
\begin{description}
\item[Distributive lattice axioms] \mbox{} \medskip
	\begin{itemize}
	\item[] $A \vdash A$ \quad 
				$A \vdash \top$ \quad $\bot \vdash A$
	\item[] $A \land B \vdash A$ \quad $A \land B \vdash B$ \quad
				 $A \vdash A \lor B$ \quad $B \vdash A \lor B$
	\item[]	 $A \land (B \lor C) \vdash (A \land B) \lor (A \land C)$ 
	\end{itemize}
	\medskip

\item[Residuation rules] \mbox{} \medskip
	\begin{itemize}
	\item[] $\dfrac{A \PR B \vdash C}{B \vdash A \LD C}$ \qquad $\dfrac{B \vdash A \LD C}{A \PR B \vdash C}$ \qquad $\dfrac{A \PR B \vdash C}{A \vdash C \RD B}$ \qquad $\dfrac{A \vdash C \RD B}{A \PR B \vdash C}$
	
	\end{itemize}
	\medskip

\item[Iteration axioms] \mbox{} \medskip

	\begin{minipage}{0.5\linewidth}
	\begin{itemize}
	\item[] $A \LDI B \land A \LDI C \vdash A \LDI (B \land C)$
	\item[] $A \LDI B \vdash A \LD B \land A \LD (A \LDI B)$
	\item[] $A \LD B \land A \LD (A \LDI B) \vdash A \LDI B$
	\end{itemize}
	\end{minipage}
	\begin{minipage}{0.5\linewidth}
	\begin{itemize}
	\item[] $B \RDI A \land C \RDI A \vdash (B \land C) \RDI A$
	\item[] $B \RDI A \vdash B \RD A \land (B \RDI A) \RD A$
	\item[] $(B \RDI A) \RD A \land B \RD A \vdash B \RDI A$
	\end{itemize}
	\end{minipage}
	\medskip
	
\item[Distributive lattice rules] \mbox{} \medskip
	\begin{itemize}
	\item[] $\dfrac{A \vdash B \quad A \vdash C}{A \vdash B \land C}$ \qquad
				$\dfrac{A \vdash C \quad B \vdash C}{A \lor B \vdash C}$\qquad
				 $\dfrac{A \vdash B \quad B \vdash C}{A \vdash C}$
	\end{itemize}
	\medskip
 	
 \item[Iteration rules] \mbox{} \medskip
  	\begin{itemize}
	\item[] $\dfrac{A \vdash B \quad C \vdash D}{B \LDI C \vdash A \LDI D}$ \qquad
	$\dfrac{A \vdash B \quad C \vdash D}{C \RDI B \vdash D \RDI A}$ \qquad
	 $\dfrac{A \vdash B \LD A}{A \vdash B \LDI A}$\qquad
	$\dfrac{A \vdash A \RD B}{A \vdash A \RDI B}$
	\end{itemize}
\end{description}
Theorems are defined as usual. We write $A \longrightarrow B$ instead of ``$A \vdash B$ is provable in $\mathit{IDFNL}$'' for the rest of the paper.

\begin{remark}\label{rem: divisions and modalities}
If the sequence of symbols ``$A \LD$'' is seen as a formula-indexed unary modality $[A]$, the special case ``$(A \lor B) \LD$'' is seen as a choice modality $[A \cup B]$ and the special case ``$A \LDI$'' is seen as a transitive-closure modality $[A^{+}]$, then some of the iteration axioms and rules correspond to variants of Segerberg's axioms for Propositional Dynamic Logic; see \cite[ch.\ 7]{Harel2000}, for example. (A similar claim applies to ``$\RD A$'', which we leave out of the present discussion.) It follows from the residuation rules that i) $[A]$ is a regular modality in the sense of \cite{Segerberg1971,Chellas1980}---$[A]B \land [A]C$ entails $[A](B \land C)$; ii) $[A]$ is a monotonic modality---if $B$ entails $C$, then $[A]B$ entails $[A]C$; iii) $[A]$ is a normal modality---$[A]\top$ is entailed by $\top$. For similar reasons, $[A^{+}]$ can also be seen as a normal, regular and monotonic modality.

The residuation rules also imply that ``choice'' $[A \cup B]$ satisfies the usual PDL reduction axiom, according to which $[A \cup B]C$ is equivalent to $[A]C \land [B]C$. The second iteration rule corresponds to the \emph{Loop Invariance Rule} of PDL according to which $A$ entails $[B^{+}]A$ in case $A$ entails $[B]A$. The second and third iteration axiom together say that $[A^{+}]B$ is a fixed point of the function $f: C \mapsto [A]B \land [A]C$.\footnote{Strictly speaking, this function should be defined on the equivalence classes of formulas.} Note that $[A^{+}]B$ is not a fixed point of $f' : C \mapsto [A]C$ (we are dealing with a version of transitive closure---but recall Remark \ref{rem: transitive closure}---not reflexive transitive closure).

As noted earlier, the dual of the first iteration axiom does not hold; this is reminiscent of the situation in PDL where, transposed to our setting, $[A^{+} \cup B^{+}]C$ is not equivalent to $[(A \cup B)^{+}]C$.
\end{remark}

Given Remark \ref{rem: divisions and modalities}, it is not surprising that completeness and decidability results concerning our logic with iterative division operators can be established using arguments similar to the proofs for Propositional Dynamic Logic. Given that we are working in a setting without Boolean negation, we find the approach of Nishimura \cite{Nishimura1982} particularly suitable. We augment our variation of Nishimura's argument with usual techniques pertaining to relational semantics of substructural logics \cite{Restall2000}.

Let $L$ be any extension of Distributive Lattice Logic. The notion of a prime $L$-theory is defined as expected (a set of formulas closed under conjunction and provable sequents that has the disjunction property). An independent $L$-pair is an ordered pair of sets of formulas $\Gamma, \Delta$ such that there are no finite $\Gamma' \subseteq \Gamma$, $\Delta' \subseteq \Delta$ such that $\bigwedge \Gamma' \longrightarrow \bigvee \Delta'$ (in $L$). Recall that $\bigwedge \emptyset := \top$ and $\bigvee \emptyset := \bot$. We will rely on the following well-known result:

\begin{lemma}[Pair Extension]\label{lem: pair extension}
Let $L$ be an extension of the Distributive Lattice Logic and let $\langle \Gamma, \Delta \rangle$ be an independent $L$-pair. Assume that $\Phi$ is a set of formulas such that $\Gamma \cup \Delta \subseteq \Phi$. There is an independent $L$-pair $\langle \Gamma', \Delta'\rangle$ such that $\Gamma \subseteq \Gamma'$, $\Delta \subseteq \Delta'$ and $\Gamma' \cup \Delta' = \Phi$.
\end{lemma} 
\begin{proof}
Essentially \cite[pp.\ 92--95]{Restall2000}.
\end{proof}

\begin{corollary}
If $L$ is an extension of Distributive Lattice Logic and $\langle \Gamma, \Delta\rangle$ is an independent $L$-pair, then there is a prime $L$-theory $\Sigma$ such that $\Gamma \subseteq \Sigma$ and $\Delta \cap \Sigma = \emptyset$.
\end{corollary}
\begin{proof}
Lemma \ref{lem: pair extension} in the case $\Phi$ is the whole language; in that case $\Gamma'$ is a prime theory (see \cite[Lemma 5.16]{Restall2000}).
\end{proof}

The \emph{restriction} of $\langle \Gamma, \Delta\rangle$ to $\Sigma$ is $\langle \Gamma \cap \Sigma, \Delta \cap \Sigma\rangle$. We say that $\langle \Gamma, \Delta\rangle$ is a \emph{full} pair iff $\Gamma \cup \Delta$ is the set of all formulas.

\begin{definition}[Closure]
A set of formulas $\Phi$ is the closure of a set of formulas $\Phi'$ iff it is the smallest set of formulas such that
\begin{enumerate}
\item $\Phi' \subseteq \Phi$;
\item $\top, \bot \in \Phi$;
\item $\Phi$ is closed under subformulas;
\item If $A \LDI B \in \Phi$, then $A \LD B \in \Phi$;
\item If $B \RDI A \in \Phi$, then $B \RD A \in \Phi$. 
\end{enumerate}
If $\Phi$ is the closure of $\Phi$, then we say that $\Phi$ is closed.
\end{definition}

\begin{lemma}
The closure of any finite set is finite.
\end{lemma}

\begin{definition}[Finite canonical model]
Let $\Phi$ be a finite closed set. The model $M_{\Phi}$ is defined as follows.
\begin{itemize}
\item $S_{\Phi}$ is the set of independent $\mathit{IDFNL}$-pairs $a = \langle a_{in}, a_{out} \rangle$ such that $a_{in} \cup a_{out} = \Phi$ (note that $a_{in} \cap a_{out} = \emptyset$ by definition of independent pair); we often write ``$A \in a$'' instead of ``$A \in a_{in}$'' and ``$a \vdash A$'' instead of ``$\bigwedge a_{in} \longrightarrow A$'';
\item $R_{\Phi}abc$ iff, for all $A,B \in \Phi$, if $a \vdash A$ and $b \vdash A \LD B$, then $c \vdash B$ (note that, if $C \in \Phi$, then $a \vdash C$ iff $C \in a$);
\item if $p \in \Phi$, then $a \in V_{\Phi}(p)$ iff $p \in a$; otherwise $V_{\Phi}(p) = \emptyset$. 
 \end{itemize}
The satisfaction relation $\vDash_\Phi$ is defined in the usual manner.
\end{definition}

We note that the restriction to \emph{finite} closed sets is required by our strategy in proving the next lemma for $\LDI$ and $\RDI$.  

\begin{lemma}[Truth Lemma]\label{lem: truth lemma}
For all $M_{\Phi}$, $E \in \Phi$ and $z \in M_{\Phi}$, $E \in z$ iff $z \vDash_{\Phi} E$.
\end{lemma}
\begin{proof}
Induction on the complexity of $E$. The base case holds by definition and the cases for constants and lattice connectives are easily established using the Distributive lattice axioms and rules. The cases for the Lambek connectives are established using standard arguments \cite{Restall2000}; we give the ones for $\LD$ and $\PR$ in full, just in case. We mostly omit the subscript $\Phi$ in the rest of the proof.

\underline{$E = A \LD B$.} $A \LD B \in b$ implies $b \vDash A \LD B$ by the definition of $R_{\Phi}$. Conversely, if $A \LD B \notin b$, then we reason as follows. Firstly, $\langle \{ C \mid b \vdash A \LD C \}, \{ B \} \rangle$ is an independent pair and so there is, by the Pair Extension Lemma \ref{lem: pair extension}, a full independent pair $\langle \Gamma, \Delta\rangle$, such that $\Gamma$ extends $\{ C \mid b \vdash A \LD C \}$ and $B \in \Delta$. The restriction $c = \langle c_{in}, c_{out}\rangle$ of $\langle\Gamma, \Delta\rangle$ to $\Phi$ is clearly in $S_{\Phi}$. Next, take the pair $\langle \{ A \}, \{ D \mid \exists C \notin \Gamma : b \vdash D \LD C \}\rangle$; this pair is independent, for otherwise $b \vdash A \LD \bigvee C_i$ for some disjunction of $C_i \notin \Gamma$ and thus some $C_i \in \Gamma$ by the construction of $\Gamma$, leading to a contradiction. Hence, there is a full independent pair $\langle \Sigma, \Theta\rangle$ such that $A \in \Sigma$ and $\Theta$ extends $\{ D \mid \exists C \notin \Gamma : b \vdash D \LD C \}$. The restriction $a = \langle a_{in}, a_{out}\rangle$ of $\langle \Sigma, \Theta\rangle$ to $\Phi$ is in $S_{\Phi}$. Moreover, $R_{\Phi}abc$; for take $D, C \in \Phi$ such that $b \vdash D \LD C$ and $c \not\vdash C$. The latter means that $C \notin \Gamma$, but then $a \not\vdash D$ by the construction of $a$. The case \underline{$E = B \RD A$} is established similarly.

\underline{$E = A \PR B$.} If $b \vdash B$, then $b \vdash A \LD (A \PR B)$ since $B \longrightarrow A \LD (A \PR B)$ by residuation. If also $a \vdash A$, then $Rabc$ implies $c \vdash A \PR B$ by the definition of $R_{\Phi}$. Hence, $c \vDash A \PR B$ implies $c \vdash A \PR B$. The converse implication is established as follows. Assume that $c \vdash A \PR B$. Then $c$ is extended by a full independent pair $\langle\Gamma, \Delta\rangle$. The pair $\langle \{ B \}, \{ C \mid \Gamma \not\vdash A  \PR C \}\rangle$ is easily shown to be independent; thus it is extended by some full independent pair $\langle \Sigma, \Theta\rangle$ with a restriction $b$ to $\Phi$. The pair $\langle \{A\}, \{ C \mid \exists D \in \Sigma : \Gamma \not\vdash C \PR D\}\rangle$ is also easily shown to be independent and thus extended by a full independent pair with a restriction $a$ to $\Phi$. It is clear that $a \vdash A$ and $b \vdash B$. It remains to show that $R_{\Phi}abc$. If $a \vdash C$ and $b \vdash C \LD D$, for some $C,D \in \Phi$, then $\Gamma \vdash C \PR (C \LD D)$ by the construction of $a$ and $b$. Hence, $\Gamma \vdash D$ and so $c \vdash D$. Consequently, $c \vdash A \PR B$ implies $c \vDash A \PR B$. 

\underline{$E = A \LDI B$.} First we prove that if $A \LDI B \in a$, $R \overleftarrow{x} a b$ for some $\bar{x}$ of length $n$ and $A \in x_i$ for all $i \in \{ 1, \ldots, n \}$, then $B \in b$. If $\bar{x} = \langle x\rangle$, then we have $Rxab$ and, by the second $\LDI$-axiom and the definition of closure, $A \LD B \in a$. Hence, $B \in b$ by the definition of the canonical $R$. Now assume that $\bar{x}$ of length $n \geq 2$. The assumption $R \overleftarrow{x} ab$ means that there are $y_1, \ldots, y_{n-1}$ such that 
\[
Rx_1ay_1 \land \bigwedge_{1 \leq i \leq n-2} Rx_{i+1} y_i y_{i+1} \land Rx_{n}y_{n-1}b
\]
Since $A \LDI B \longrightarrow A \LD (A \LDI B)$ by the second $\LDI$-axiom and both $A$ and $A \LDI B$ are in $\Phi$, we have $A \LDI B \in y_i$ for all $i \in \{ 1, \ldots, n-1 \}$. Since $A \LDI B \longrightarrow A \LD B$ by the second $\LDI$-axiom, $A \LD B \in y_{n-1}$ and so $B \in b$ by the definition of $R$. 

The converse claim is established as follows. Assume that $a \vDash_{\Phi} A \LDI B$. We define ($\bar{x}_n$ means that path $\bar{x}$ is of length $n$)
\begin{align*}
Y & := \Big\{ y \mid (\exists n \exists \bar{x}_n) \big( R \overleftarrow{x} ay \land (\forall i \in \{ 1, \ldots, n \})(A \in x_i) \big) \Big\} \\
\phi_Y & := \bigvee_{b \in Y} \big( \bigwedge b_{in} \big)
\end{align*}
(Note that if $Y = \emptyset$, then $\phi_Y$ is $\bot$.) We sometimes write ``$Y$'' instead of ``$\phi_Y$''and let the context disambiguate.

We first show that $\phi_Y \longrightarrow A \LD \phi_Y$ (here we need the assumption that $Y$ is finite). If not, then $\exists b \in Y$ such that $b \not\vdash A \LD Y$. It follows that $\bigwedge \{ C \mid b \vdash A \LD C \} \not\longrightarrow Y$. (Since $b \vdash A \LD \bigwedge \{ C \mid b \vdash A \LD C \}$.) So, by the Pair Extension Lemma, there is a prime theory $\Gamma$ containing all $C$ such that $b \vdash A \LD C$ and disjoint from the set of disjuncts of $\phi_Y$. Now take the independent pair $ \langle \{ A \}, \{ D \mid \exists C \notin \Gamma : b \vdash D \LD C \}\rangle$ (for the proof of independence, see the argument in case $E = A \LD B$); there is $d \in S$ such that $A \in d_{in}$ and $\{ D \mid \exists C \notin \Gamma : b \vdash D \LD C \} \cap \Phi \subseteq d_{out}$ by the Pair Extension Lemma. Define $c_{in} := \Gamma \cap \Phi$ and $c_{out}$ as the complement of $c_{in}$ relative to $\Phi$. It is easily seen that $R dbc$, for if $C \in \Phi$ and $C \notin c$, then $C \notin \Gamma$ and so, if $d \vdash D$ for some $D \in \Phi$, then $b \not\vdash D \LD C$ by the construction of $d$. Since $A \in d$ and $b \in Y$, we have $c \in Y$ and so $c \vdash Y$. But then $Y \in \Gamma$, which is impossible by the construction of $\Gamma$.

Hence, $Y \longrightarrow A \LD Y$ and so, by the second $\LDI$-rule, $Y \longrightarrow A \LDI Y$. Since all elements of $Y$ contain $B$, we have $Y \longrightarrow B$ and so $Y \longrightarrow A \LDI B$. It can be shown similarly as above that $a \vdash A \LD Y$; but since $A \LD Y \longrightarrow A \LD B \land A \LD (A \LDI B) \longrightarrow A \LDI B$, we obtain $a \vdash A \LDI B$. This means that $A \LDI B \in a$. 

The claim for \underline{$E = B \RDI A$} is established similarly.
\end{proof}

\begin{theorem}\label{thm: completeness}
$\mathit{IDFNL}$ is a sound and weakly complete axiomatization of the $\mathcal{L}_1$-theory of all Do\v{s}en frames.
\end{theorem}
\begin{proof}
Soundness is an easy exercise; completeness follows from the construction of the finite canonical model (for each $A \not\longrightarrow B$, take as $\Phi$ the closure of $\{ A, B \}$ and construct $M_{\Phi}$) and the Truth Lemma \ref{lem: truth lemma} in conjunction with the Pair Extension Lemma \ref{lem: pair extension} (if  $A \not\longrightarrow B$, then there is $a$ in $M_{\Phi}$ such that $A \in a_{in}$ and $B \in a_{out}$; by the Truth Lemma, there is a state in $M_{\Phi}$ satisfying $A$ that does not satisfy $B$).
\end{proof}

\begin{theorem}
The $\mathcal{L}_1$-theory of all Do\v{s}en frames is decidable.
\end{theorem}
\begin{proof}
The proof of Theorem \ref{thm: completeness} entails axiomatizability and completeness with respect to a recursively enumerable class of models (finite Do\v{s}en models).
\end{proof}

\section{Conclusion}\label{sec: conclusion}
In this article we have put forward a decidable extension of the Distributive Full Non-associative Lambek Calculus that allows to reason about the effects of iterative composition (of bracketed expressions in Lambek's sense, or more generally of linguistic resources or actions). We have achieved this by extending DFNL with two primitive iterative division operators $\LDI$ and $\RDI$ and by providing a complete axiomatization and establishing the finite model property. (In fact, we have established \emph{bounded} finite model property, from which decidability follows independently of the axiomatization.)

In the rest of this section, we discuss some of our design choices and plans for future work. Firstly, the reader may wonder why we have left out the discussion of an ``iterative product'' operator with respect to which $\LDI$ and $\RDI$ would be residuated. In fact, we suspect that there is not a single such operator; technically speaking, the iterative division operators ``look at'' two different accessibility relations. This will be investigated in more detail in the future. 

Concerning our choice of the language, we note that the presence of $\top, \bot$ is motivated by technical considerations. In our proof of the Truth Lemma for the iterative division operators, the set $Y$ may be empty in which case $\phi_Y$ is not well defined without $\bot$ in the language. In order to obtain the Truth Lemma for $\bot$, we need $\top$. (When it comes to the informal interpretation of these constants, $\top$ can be read as the trivial type of linguistic resource or action---all objects are of this type; and $\bot$ can be read as the inconsistent type---no object is of this type.)

Our work is motivated by the interest of studying iterated composition in non-associative settings. We have presented the basic logic of such settings, but it is interesting  also to look at its extensions, especially those where composition has some ``natural'' properties. Some of these are familiar from the literature on the Lambek Calculus. For instance, it is not hard to show that the extension of $\mathit{IDFNL}$ with the Weak contraction axiom
\[
A \land A \LD B \vdash B
\]
is a sound and weakly complete axiomatization of the theory of Do\v{s}en frames satisfying reflexivity $Rsss$ and that this theory is decidable. This frame condition reflects the idea that composing an object with itself may always result in the object itself. This is again plausible in some cases, but not in general (action composition is a counterexample). 

Similarly, the extension of $\mathit{IDFNL}$ with the Weak commutativity axiom \[ A \PR B \vdash B \PR A \] is a sound and weakly complete axiomatization of the theory of all frames such that $Rstu $ implies $ Rtsu$; this theory is decidable as well. The frame condition is reflecting the idea that the order of composition is immaterial when it comes to the output. (This may be the case for some special cases of composition, but not in general---think of action composition as an example.) 

Some natural assumptions, however, go beyond the limits of the present framework. For instance, if composition is thought of as \emph{information update} and formulas are seen as expressing information \emph{entailed by} states, then it is plausible to assume the Success axiom
\[
A \LD B \vdash A \LD (A \land B)
\]
expressing the notion that ``update with $A$'' results in a state entailing $A$. Many notions of update studied in the epistemic logic literature share this feature (e.g.\ public announcements or belief revision), but it can be shown that the Success axiom is \emph{not canonical} when the present notion of a canonical model is assumed. To repair this, we would have to add a partial information order $\leq$ (as in the Routley--Meyer semantics for relevant logics \cite{Restall2000}) and assume that $Rstu \implies s \leq u$. This can be done, but it would make the semantics more complicated.

There are also other properties of composition not usually studied in the context of the Lambek Calculus (or substructural logics in general), but natural when the present setting is seen as a generalization of various logics of epistemic dynamics. For instance, in Public Announcement Logic it is assumed that if it is possible to compose information $s$ as input with information $t$ (to publicly announce $s$ in the context of $t$), then $s$ has to be ``true'' with respect to $t$. This may be seen as corresponding to $Rstu \implies s \leq t$. For more on frame properties corresponding to various notions of information update (in the setting of normal modal logic), see \cite{Benthem2014,Holliday2012}.

In the future, we would like to take a closer look at the variants of our basic logic discussed above.

\end{document}